\documentclass[11pt,a4paper,twocolumn]{article}

\usepackage[utf8]{inputenc}
\usepackage{times}
\usepackage[colorlinks=false, pdfborder={0 0 0}]{hyperref}
\usepackage{breakurl}

\usepackage{amssymb}
\usepackage{amsmath}
\usepackage{amsthm}
\usepackage{listings}

\newtheorem{lemma}{Lemma}
\newtheorem{theorem}{Theorem}

\DeclareMathOperator{\sign}{sign}

\allowdisplaybreaks

\usepackage[font=small,margin=1ex,labelfont=bf]{caption}

\usepackage{subfigure}

\usepackage[left=2.0cm,right=2.0cm,top=3.0cm,bottom=3.0cm]{geometry}

\usepackage{graphicx}

\usepackage{fancyhdr}
\fancypagestyle{firstpage}{
  \lhead{Published in the Journal of Statistical Theory and Practice 18, 51 (2024)}
  \rhead{}
  \cfoot{\footnotesize{\em This work is licensed under the Creative Commons Attribution 4.0 license.\\
  For the Official Open Access publication, see \url{https://doi.org/10.1007/s42519-024-00399-y}}}
  
}

\title{\vspace*{-10mm}A Simple Bias Reduction for Chatterjee's Correlation}
\author{
Christoph Dalitz, Juliane Arning, Steffen Goebbels\\
Hochschule Niederrhein\\
Institut f\"ur Mustererkennung\\
Reinarzstr. 49, 47805 Krefeld\\
{\tt christoph.dalitz{@}hsnr.de}
}
\date{}

\begin{document}

\renewcommand{\labelenumi}{\arabic{enumi})}

\twocolumn[
  \begin{@twocolumnfalse}
    \maketitle
\begin{abstract}
Chatterjee's rank correlation coefficient $\xi_n$ is an empirical index for detecting functional dependencies between two variables $X$ and $Y$. It is an estimator for a theoretical quantity $\xi$ that is zero for independence and one if $Y$ is a measurable function of $X$. Based on an equivalent characterization of sorted numbers, we derive an upper bound for $\xi_n$ and suggest a simple normalization aimed at reducing its bias for small sample size $n$. In Monte Carlo simulations of various models, the normalization reduced the bias in all cases. The mean squared error was reduced, too, for values of $\xi$ greater than about 0.4. Moreover, we observed that non-parametric confidence intervals for $\xi$ based on bootstrapping $\xi_n$ in the usual n-out-of-n way have a coverage probability close to zero. This is remedied by an m-out-of-n bootstrap without replacement in combination with our normalization method.
\end{abstract}
\vspace*{2ex}

  \end{@twocolumnfalse}
  ]

\thispagestyle{firstpage}

\section{Introduction}
\label{sec:intro}
The common correlation coefficients by Pearson, Spearman, or Kendall can only be used for identifying simple relationships between two random variables $X$ and $Y$. Pearson's correlation coefficient approaches one only for a linear relationship, whereas Spearman's $\rho$ and Kendall's $\tau$ are based on ranks and thus less susceptible to the particular functional form of the relationship as long as it is monotone. For non-monotonous relationships, however, e.g.~for $X$ uniformly distributed in $[-1,1]$ and $Y=X^2$, all three correlation coefficients can be zero on average.

Recently, Chatterjee has proposed a correlation coefficient $\xi_n$ that overcomes this restriction. He has shown that, as the data set size $n$ increases, this correlation coefficient converges to a theoretical quantity $\xi(X,Y)$ that is one, if $Y$ is a measurable function of $X$, and zero, if $X$ and $Y$ are independent \cite{chatterjee21}. The coefficient was subsequently generalized by Azadkia \& Chatterjee \cite{azadkia21} to a measure for conditional dependency and utilized in a model-free method for feature selection in machine learning (``FOCI'').

In the present article, we propose a simple rescaling of $\xi_n$ by dividing it through its upper bound in order to reduce its bias. In a series of Monte Carlo simulations, we measure the effect of this rescaling on the bias and mean squared error (MSE) of the estimator for $\xi$. Moreover, we consider the problem of constructing confidence intervals for $\xi$. As $\xi_n$ is not a maximum likelihood estimator, neither of the methods for constructing confidence intervals based on the likelihood function, such as the profile likelihood method \cite{fischer21} or variance estimation from its Hessian matrix \cite{greene00}, are applicable.

It is thus necessary to resort to resampling methods like the bootstrap \cite{efron79,diciccio96,davison97}, which can be used to estimate parameters of a theoretically known asymptotic distribution (typically the normal distribution) or to directly estimate non-parametric confidence intervals. Even if the asymptotic behavior of an estimator is known, non-parametric intervals might be preferable due to small sample sizes which make the applicability of the asymptotic theory dubious \cite{dey16}. The usual bootstrap is based on repeated drawing with replacement (n-out-of-n bootstrap), but there are situations when this method fails \cite{abadie08,lin23}. In such cases, the m-out-of-n bootstrap is an alternative, because it works under weaker conditions \cite{politis94,bickel97}.

As Lin \& Han \cite{lin23} have demonstrated the failure of the n-out-of-n bootstrap for estimating the variance of $\xi_n$, Dette \& Kroll \cite{dette23} proposed to estimate the variance of $\xi_n$ with an m-out-of-n bootstrap and to assume normality for the construction of a confidence interval. In the case of independent $X$ and $Y$, the asymptotic validity of the normality assumption was already proven by Chatterjee \cite{chatterjee21}, and for {\em continuous} $X$ and $Y$ it was subsequently proven by Lin \& Han \cite{lin22}. An alternative approach without assuming normality are non-parametric bootstrap confidence intervals. For continuous $Y$, however, we observed that the resulting intervals of the n-out-of-n bootstrap had a coverage probability close to zero and did not even include the point estimate. We therefore resorted to the non-parametric m-out-of-n bootstrap according to Politis \& Romano \cite{politis94}. For small sample size $n$, these intervals had a higher coverage probability in most of our simulations than the intervals based on the normality assumption, but the coverage probability of the latter intervals approached the nominal value more quickly, which makes them more recommendable for moderate and large sample sizes.

Measuring bias, MSE, and coverage probability requires a computation of the true value of $\xi$. In order to be able to compute this value, we have transformed Chatterjee's formula into an equivalent expression based on the conditional probability $P(Y\geq t|X=x)$, that further simplifies considerably for a continuous variable $Y$ (see Theorem \ref{theorem:xicondprob.cont} below).

This article is organized as follows: In Section \ref{sec:xicor}, we summarize the definition of Chatterjee's correlation $\xi_n$ and its relationship to a correlation index by Dette et al.~\cite{dette13}, in Section \ref{sec:normalization} we derive an upper bound for $\xi_n$ and suggest a normalization, in Section \ref{sec:simulation} we examine the effect of the normalization on bias and MSE, and in Section \ref{sec:coverage}, we examine its effect on the coverage probability of a non-parametric m-out-of-n bootstrap confidence interval and compare this approach with the confidence interval by Dette \& Kroll \cite{dette23}. Technical details about computing the limiting value for $\xi_n$ have been moved to the appendix.

\section{Chatterjee's correlation}
\label{sec:xicor}
Let $(X,Y)$ be a pair of random variables, where $Y$ is not a constant. Chatterjee's correlation coefficient is based on the ranks of the values observed for $Y$. Let $(x_1,y_1),\ldots,(x_n,y_n)$ be $n$ pairs of independently drawn sample values for $(X,Y)$, for which not all $y_i$ are identical, and which are ordered such that $x_1\leq x_2\leq\ldots\leq x_n$. Let $r_i$ be the rank of $y_i$, i.e., the number of values in $\vec{y}=(y_1,\ldots,y_n)$ that are less or equal to $y_i$. Additionally, let $l_i$ be the number of values in $\vec{y}$ that are greater or equal to $y_i$. Then Chatterjee's rank correlation coefficient is defined as
\begin{equation}
  \label{eq:xin}
  \xi_n(\vec{x},\vec{y}) = 1 - \frac{n\sum_{i=1}^{n-1} |r_{i+1} - r_i|}{2\sum_{i=1}^n l_i(n-l_i)}
\end{equation}
Note that in case of equal $x_i=x_j$ corresponding to different $y_i\neq y_j$, the value of $\xi_n$ depends on the tie breaking, which should be done at random according to Chatterjee. If all $x_i$ are different, this can not occur. Moreover, in the case of no ties among the $y_i$, the denominator is $n(n^2-1)/3$ and Eq.~(\ref{eq:xin}) simplifies to
\begin{equation}
  \label{eq:xin:noties}
  \xi_n(\vec{x},\vec{y}) = 1 - \frac{3\sum_{i=1}^{n-1} |r_{i+1} - r_i|}{n^2 -1}
\end{equation}

Chatterjee has shown that, if $(x_i,y_i)$ are i.i.d. pairs drawn according to the law of $(X,Y)$ and $Y$ is not almost surely constant, then $\xi_n(\vec{x},\vec{y})$ converges for $n\to\infty$ almost surely to the quantity
\begin{equation}
  \label{eq:xiorig}
  \xi(X,Y) = \frac{\int \operatorname{Var}\Big(E(1_{\{Y\geq t\}}|X)\Big) d\mu(t)}{\int \operatorname{Var}\Big(1_{\{Y\geq t\}}\Big) d\mu(t)}
\end{equation}
where the variable $t$ in the variance is treated as a constant parameter, and $\mu$ is the probability distribution of $Y$. If $X$ and $Y$ are independent, the random variable $E(1_{\{Y\geq t\}}|X)$ is equal to $P(Y\geq t)$ and its variance is zero. If $Y$ is a measurable function of $X$, then $Y=f(X)$ is constant for a given event $X=x$, and $E(1_{\{Y\geq t\}} | X)=1_{\{Y\geq t\}}$, which means that numerator and denominator in (\ref{eq:xiorig}) are identical.

In order to better understand expression (\ref{eq:xiorig}) and make it easier to compute in practical use cases, we can rewrite it in terms of the conditional probability $P(Y\geq t|X=x)$ and the unconditional probability
\begin{equation}
  \label{eq:totprob}
  P(Y\geq t) = \int P(Y\geq t\mid X=x) \, d\lambda(x)
\end{equation}
where $\lambda$ is the probability distribution of $X$.
First note that $E(1_{\{Y\geq t\}}|X)$ is just $P(Y\geq t|X)$. Further note that $1_{\{Y\geq t\}}$ is a Bernoulli variable which can take only two values:
\begin{displaymath}
  1_{\{Y\geq t\}} = \left\{\begin{array}{ll}
  1 & \mbox{ with probability }P(Y\geq t) \\
  0 & \mbox{ with probability }1-P(Y\geq t)
  \end{array}\right.
\end{displaymath}
The variance of a Bernoulli variable $B$ with $P(B=1)=p$ is $\operatorname{Var}(B)=p(1-p)$, and, with the use of the formula $\mbox{Var}(Z)=E(Z^2)-E(Z)^2$ for the variance in the numerator, Eq.~(\ref{eq:xiorig}) can be written as
\begin{align}
  \label{eq:xicondprob}
  \xi&(X,Y) = \\
  & \frac{\int\! \Big(\!\int\! P(Y\!\geq\! t|X\!=\!x)^2 d\lambda(x) - P(Y\!\geq\! t)^2\Big)  d\mu(t)}{\int P(Y\!\geq\! t)\Big(1-P(Y\!\geq\! t)\Big) \, d\mu(t)} \nonumber
\end{align}
where $\lambda$ is the probability distribution of $X$ and $\mu$ is the probability distribution of $Y$.

For a continuous random variable $Y$, two of the integrals can be solved symbolically, which results in a simplified formula for $\xi$ which was already proposed by Dette et al.~\cite{dette13} as a dependency measure\footnote{Dette et al.~used ``$\leq$'' instead of ``$\geq$'', but, for continuous $Y$, this is equivalent as can be shown with the law of total probability.}.

\begin{theorem}
  \label{theorem:xicondprob.cont}
  For a continuous random variable $Y$,
  $$\xi(X,Y) = 
  6\int \!\!\!\int\! P(Y\!\geq\! t|X\!=\!x)^2 d\lambda(x)\, d\mu(t) - 2$$
\end{theorem}
\begin{proof}
For continuous $Y$, we can write $d\mu(t)=f(t)dt$ with $f(t)=-\frac{d}{dt}P(Y\geq t)$. This yields
\begin{align}
  \label{eq:P1}
  \int P(Y\!\geq\! t)^2  d\mu(t) &= -\frac{1}{3}P(Y\!\geq\! t)^3\Big|_{-\infty}^{\infty} = \frac{1}{3} \\
  \label{eq:P2}
  \int P(Y\!\geq\! t)  d\mu(t) &= -\frac{1}{2}P(Y\!\geq\! t)^2\Big|_{-\infty}^{\infty} = \frac{1}{2}
\end{align}
Inserting these values into (\ref{eq:xicondprob}) yields the theorem. Note that equalities (\ref{eq:P1}) and (\ref{eq:P2}) no longer hold for discrete random variables: In this case additional positive sums occur on the right hand side and the expressions are strictly greater than $1/3$, or $1/2$, respectively.
\end{proof}

The interesting property of $\xi$ with respect to dependency between $Y$ and $X$ can be readily seen in the expression in the above theorem. For independence of $X$ and $Y$, it is $P(Y\!\geq\! t|X=x)^2= P(Y\!\geq\! t)^2$ and the integral evaluates to $1/3$. If $Y$ is a function of $X$, i.e., $Y=f(X)$, $P(Y\!\geq\! t|X=x)^2$ is one if $f(x)\geq t$ and zero if $f(x)<t$, which means that the inner integral becomes $P(Y\geq t)$ and the outer integral evaluates to $1/2$.

\section{Normalization of $\xi_n$}
\label{sec:normalization}
Although Chatterjee has proven in \cite{chatterjee21} that the coefficient $\xi_n$ converges almost surely to $\xi$, the average value of $\xi_n$ can be considerably different from the asymptotic value $\xi$ for small sample sizes. Consider, e.g., the extreme case of $X=Y$. In this case Pearson's correlation is always one, whereas $\xi_n$ yields a considerably different value as the following sample R session with the use of the package {\em XICOR} \cite{xicor_r} demonstrates:
\begin{quote}
\begin{verbatim}
> x <- runif(20)
> cor(x,x)
[1] 1
> xicor(x,x)
[1] 0.8571429
\end{verbatim}
\end{quote}
This is because the maximum possible value of Eq.~(\ref{eq:xin}) is less than one and this upper bound decreases as $n$ becomes smaller. To understand what the maximum value of $\xi_n(\vec{x},\vec{y})$ is over all possible permutations of the vector $\vec{x}$, we need the following

\begin{lemma}
  \label{lemma:1}
  For real numbers $x_1,\ldots,x_n$, it is
  $$\sum_{k=1}^{n-1} | x_{k+1} - x_k| \geq x_{max} - x_{min}$$
  where $x_{max}=\max\{x_1,\ldots,x_n\}$ and $x_{min}=\min\{x_1,\ldots,x_n\}$. Equality holds if and only if the $x_i$ are sorted in ascending or descending order.
\end{lemma}
\begin{proof}
  If $x_{min}=x_{max}$, the assertions are obvious. Let us assume that $i$ and $j$, $i\neq j$, are the indices with $x_i=x_{max}$ and $x_j=x_{min}$ and that, without loss of generality\footnote{\label{note:1}In the other case, we simply invert all indices $k\to n-k$.}, $i>j$. Then it is
  \begin{align*}
    x_{max}- & x_{min} = | x_i - x_j | = | (x_i -x_{i-1}) \\
    & \quad + (x_{i-1} -x_{i-2}) \ldots + (x_{j+1}-x_j) | \\
    &\leq \sum_{k=j}^{i-1} |x_{k+1}-x_k| \leq \sum_{k=1}^{n-1} |x_{k+1}-x_k| \\
  \end{align*}
  If the $x_i$ are sorted, equality holds, because for $x_{min}=x_1\leq x_2\leq\ldots\leq x_n=x_{max}$ it is\footnotemark[1]
  $$\sum_{k=1}^{n-1}|x_{k+1} - x_k| = \sum_{k=1}^{n-1}(x_{k+1} - x_k) = x_n - x_1$$
  In order to prove that the reverse holds too, i.e., that if the sum takes its minimum, the $x_i$ are sorted, let us assume that the sum takes the minimum value $x_{max}-x_{min}$, that $x_i=x_{min}$ and $x_j=x_{max}$, and that\footnote{In this case we will see that the sequence has ascending order. If $j>j$, we will have descending order and the proof remains the same, but with all indices inverted.} $i<j$. Because of
  \begin{align*}
    x_{max}-x_{min} &=\sum_{k=1}^{n-1}|x_{k+1}-x_k|\\
    &\geq \sum_{k=i}^{j-1}|x_{k+1}-x_k| \geq x_{max} - x_{min}
  \end{align*}
  equality holds and $x_1=x_2=\ldots=x_i$ and $x_j=x_{j+1}=\ldots =x_n$. Let us assume that the sequence is not sorted in ascending order, which means that there are indices $i',j'$ with $i\leq i'<j' \leq j$ and $x_{i'}>x_{j'}$. Applying the already proven inequality to different parts of the sum, we obtain
  \begin{align*}
    x_{max}-&x_{min} = \sum_{k=i}^{j-1} | x_{k+1}-x_k|
    = \underbrace{\sum_{k=i}^{i'-1} | x_{k+1}-x_k|}_{\geq \;x_{i'} - x_{min}} \\
    &\quad + \underbrace{\sum_{k=i'}^{j'-1} | x_{k+1}-x_k|}_{\geq\; x_{i'} - x_{j'}} + \underbrace{\sum_{k=j'}^{j-1} | x_{k+1}-x_k|}_{\geq\; x_{max} - x_{j'}} \\
    &\geq x_{max} - x_{min} + 2(x_{i'} - x_{j'})
  \end{align*}
  According to the assumption $x_{i'}>x_{j'}$, it is $2(x_{i'} - x_{j'})>0$, which is a contradiction, and the sequence must be in ascending order.
\end{proof}
With Lemma \ref{lemma:1}, it is easy to show that, for given values $\vec{x}$ and $\vec{y}$,  $\xi_n$ takes its maximum value if the $r_i$ in Eq.~(\ref{eq:xin}) are sorted.

\begin{theorem}
  \label{theorem:maxxi}
  Let $\xi_n$ be defined as in Eq.~(\ref{eq:xin}). Then
  $$\xi_n(\vec{x},\vec{y}) \leq \xi_n(\vec{y},\vec{y})$$
  where equality only holds for permutations of $\vec{x}$ that sort the $y_i$ in $(x_i,y_i)_{i=1}^n$ in ascending (or descending) order. The upper bound on the right is
  $$\xi_n(\vec{y},\vec{y}) =
  1 - \frac{n(n - r_{min})}{2\sum_{i=1}^n l_i(n-l_i)}$$
  where $r_{min}$ is the absolute frequency of the smallest value in $\vec{y}$. For $n>2$, it is $\xi_n(\vec{y},\vec{y})>0$. When all $y_i$ are different, the upper bound is $(n-2)/(n+1)$.
\end{theorem}
\begin{proof}
  First note that the denominator in Eq.~(\ref{eq:xin}) only depends on the values $\vec{y}$ and is unaffected by a permutation of $\vec{x}$. Thus, $\xi_n$ is maximal if $\sum_i|r_{i+1}-r_i|$ is minimal. If the $r_i$ are in ascending or descending order, the value of this sum is $r_{max}-r_{min}=n-r_{min}$, which is the smallest possible value of the sum for all permutations of the $r_i$ according to Lemma \ref{lemma:1}.

  To see that the upper bound is strictly positive for $n>2$, first note that, for $x_i=x_{min}$, it is $l_i(n-l_i)=n(n-n)=0$. As the function $f(x)=x(n-x)$ takes its maximum for $n=n/2$ and decreases with increasing distance $|x-n/2|$, the remaining $n-r_{min}$ summands in $\sum_{i=1}^n l_i(n-l_i)$ are greater than the smallest possible value $1\cdot(n-1)$, which yields
  \begin{align*}
    \sum_{i=1}^n l_i(n &- l_i) \geq (n-r_{min})(n-1)
  \end{align*}
  and thus
  \begin{align*}
    \xi_n(\vec{y},\vec{y}) &\geq 1- \frac{n(n-r_{min})}{2(n-r_{min})(n-1)} = 1-\frac{n}{2(n-1)} 
  \end{align*}
\end{proof}

As the maximum value for $\xi_n$ according to Theorem \ref{theorem:maxxi} is less than one, it must be a biased estimator in case of $\xi(X,Y)=1$, i.e., a strict dependency between $X$ and $Y$. It is thus natural to rescale $\xi_n$ to the range $[-1,1]$ by normalizing it with its maximum value:
\begin{equation}
  \label{eq:xi'}
  \xi'_{n}(\vec{x},\vec{y}) = \max\left\{-1, \frac{\xi_n(\vec{x},\vec{y})}{\xi_n(\vec{y},\vec{y})}\right\}
\end{equation}
\paragraph{Remark:} The cutoff at $-1$ has no effect if there are no ties among the $y_i$. In this case, Chatterjee already observed in \cite{chatterjee21} that
$\xi_n$ takes its minimum when when the top $n/2$ values of $y_i$ are placed alternately with the bottom $n/2$ values. This results in a lower bound
\begin{align*}
  \xi_n(\vec{x}, \vec{y}) &\geq 1-\frac{3n}{2(n+1)} \\
  &\geq -\frac{n-2}{n+1} = -\xi_n(\vec{y}, \vec{y})
\end{align*}
In case of ties, however, there are rare situations for which the ratio in (\ref{eq:xi'}) is less than one. An example is the combination $\vec{x}=(1,2,\ldots,n)$ and $\vec{y}=(0,1,0,1,\ldots,1,0)$ for odd\footnote{For even $n$, we have not found such an example.} $n$, which leads to
$$\xi_n(\vec{x}, \vec{y}) = 1-\frac{2n}{n+1} <
-\xi_n(\vec{y}, \vec{y}) = -\left(1-\frac{2n}{n^2-1}\right)$$
In our simulated models, such a situation never occurred and the cutoff was not necessary. We nevertheless leave it in the definition of $\xi_n'$ because it makes it trivial to prove the following theorem:

\begin{theorem}
  \label{theorem:xi'}
  The normalized estimator $\xi_n'$ as defined in Eq.~(\ref{eq:xi'}) has the following properties:
  \begin{enumerate}
  \item[(a)] $\xi_n'$ is an asymptotically unbiased estimator for $\xi$.
  \item[(b)] $|\xi_n'|\leq 1$
  \item[(c)] $\sign(\xi_n') = \sign(\xi_n)$
  \item[(d)] For $\xi_n>=0$, it is $\xi_n'\geq \xi_n$
  \end{enumerate}
\end{theorem}
\begin{proof} (b)--(d) follow directly from Theorem \ref{theorem:maxxi}. To see (a), remember that Chatterjee has proven in \cite{chatterjee21} that, for $n\to\infty$, $\xi_n$ converges almost surely to $\xi$. Moreover it is $\xi(Y,Y)=1$. The numerator in (\ref{eq:xi'}) thus converges almost surely to $\xi(X,Y)$ and the denominator almost surely to one. As the union of two null sets is a null set, too, the ratio converges almost surely towards $\xi(X,Y)$. According to (b), $\xi_n'$ is bounded by a constant, and it follows from Lebesgue's dominated convergence theorem that the convergence holds in expectation, too.
\end{proof}

The normalized quantity $\xi_n'$ has been constructed in such a way that it has a smaller bias than $\xi_n$ if $\xi=1$, and in this case the mean squared error will be smaller, too. This does not necessarily hold if the true value $\xi$ is less than one; in the special case of independent $X$ and $Y$, e.g., the MSE of $\xi'_n$ will be greater because $\xi=0$ and $|\xi'_n|\geq |\xi_n|$. The effect of the normalization (\ref{eq:xi'}) on the coverage probability of confidence intervals is yet another question that needs to be addressed. We therefore made extensive Monte Carlo simulations to measure these effects for different data generation processes.

\section{Effect of normalization}
\label{sec:simulation}
For ten different models for the relationship between $X$ and $Y$ and for a range of data set sizes $n$, we have simulated $N=10^6$ data sets and compared the following quality indices for the estimators $\xi_n$ and $\xi'_n$:
\begin{itemize}
\item bias
\item mean squared error (MSE)
\item coverage probability of a $90\%$ confidence interval
\end{itemize}
The computation of $\xi_n$ has been done with the function {\em xicor} from the R package {\em XICOR} \cite{xicor_r}. All three quality indices require knowledge of the true value $\xi$, which we have computed as described in the appendix. As the construction of confidence intervals for $\xi$ turned out to be more complicated than expected, we have devoted the separate Section \ref{sec:coverage} to this problem.

\begin{figure*}
  \centering
  \subfigure[model 1: $Y\sim X + \varepsilon$ (continuous)]{
      \includegraphics[width=0.95\columnwidth]{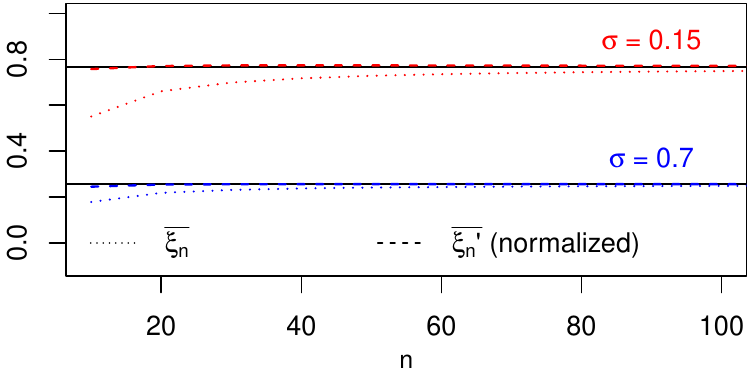}
      \label{fig:bias:model1}
  }
  \subfigure[model 3: $Y\sim \sin(2\pi X) + \varepsilon$ (continuous)]{
      \includegraphics[width=0.95\columnwidth]{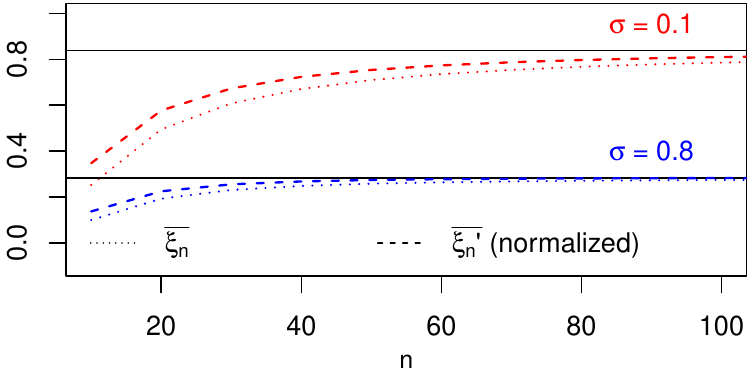}
      \label{fig:bias:model3}
  }
  \subfigure[model 4: $Y\sim XY$ (discrete)]{
      \includegraphics[width=0.95\columnwidth]{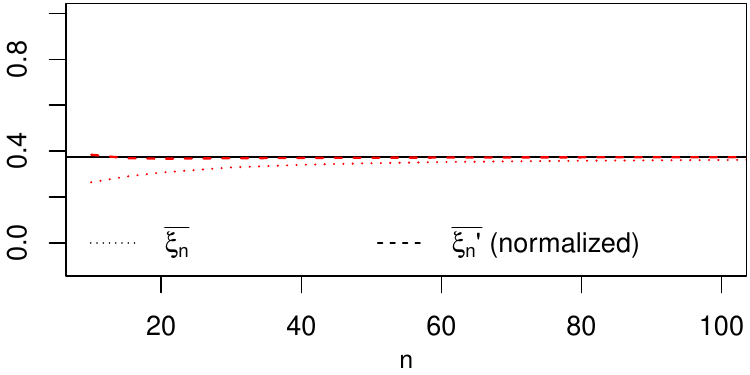}
      \label{fig:bias:model4}
  }
  \subfigure[model 6: $Y\sim X^2 + \varepsilon$ (discrete)]{
      \includegraphics[width=0.95\columnwidth]{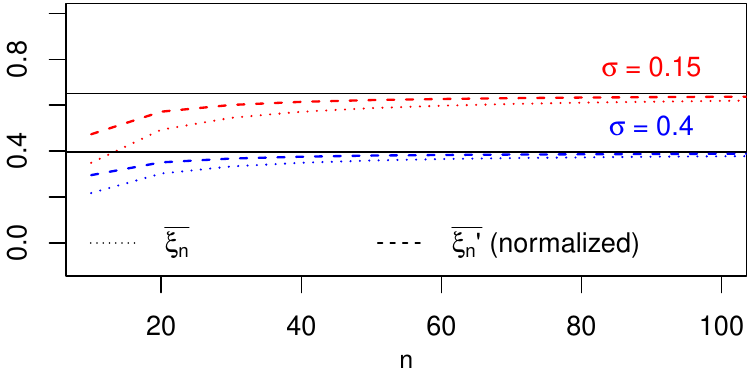}
      \label{fig:bias:model6}
  }
  \caption{\label{fig:bias} Average value of $\xi_n$ and the normalized $\xi_n'$ as a function of $n$ for different models. The solid black lines show the true value of $\xi$}
\end{figure*}

\subsection{Simulated models}
\label{sec:simulation:models}
The ten models include four models from Chatterjee's original paper, namely three continuous models \cite[Fig. 2]{chatterjee21} and the discrete model in \cite[Sec.~4.2]{chatterjee21}.

We use the notations {\em unif}$(a,b)$ for a uniform distribution between $a$ and $b$, {\em norm}$(\mu,\sigma^2)$ for a normal distribution with mean $\mu$ and variance $\sigma^2$, {\em equal}$(a,b,n)$ for a discrete uniform distribution with $n$ values between $a$ and $b$, and {\em binom}$(n,p)$ for a binomial distribution with size $n$ and probability $p$.

\paragraph{Model 1:} $\boldsymbol{Y=X+\varepsilon}$ \\
with $X\sim \mbox{unif}(-1,1)$ and $\varepsilon\sim \mbox{norm}(0,\sigma^2)$.

This is the first model of Fig.~2 in \cite{chatterjee21}.

\paragraph{Model 2:} $\boldsymbol{Y=X^2+\varepsilon}$ \\
with $X\sim \mbox{unif}(-1,1)$ and $\varepsilon\sim \mbox{norm}(0,\sigma^2)$.

This is the second model of Fig.~2 in \cite{chatterjee21}.

\paragraph{Model 3:} $\boldsymbol{Y=\sin(2\pi X)+\varepsilon}$ \\
with $X\sim \mbox{unif}(-1,1)$ and $\varepsilon\sim \mbox{norm}(0,\sigma^2)$.

This is the third model of Fig.~2 in \cite{chatterjee21}.

\paragraph{Model 4:} $\boldsymbol{Y=XZ}$ \\
with $X\sim \mbox{binom}(1,p)$ and $Z\sim \mbox{binom}(1,p')$.

This is the model used in \cite{chatterjee21} as an example for which the distribution looks normal even for dependent $X$ and $Y$. Chatterjee used the values $p=0.4$ and $p'=0.5$.

\paragraph{Model 5:} $\boldsymbol{Y=X+\varepsilon}$\\
with $X\sim \mbox{equal}(m,-1,1)$ and $\varepsilon\sim -\sigma\sqrt{m'} + \frac{2\sigma}{\sqrt{m'}}\mbox{binom}(m',0.5)$.

This is a discrete version of Model 1. $\sigma$ is the standard deviation of the noise $\varepsilon$. For $m$ and $m'$, we have used the fixed values $m=6$, $m'=2$.

\paragraph{Model 6:} $\boldsymbol{Y=X^2+\varepsilon}$\\
with $X\sim \mbox{equal}(m,-1,1)$ and $\varepsilon\sim -\sigma\sqrt{m'} + \frac{2\sigma}{\sqrt{m'}}\mbox{binom}(m',0.5)$.

This is the discrete version of Model 2, where the noise is simulated by a multiple fair coin toss. The shift and the scaling factor are chosen in such a way, that $\varepsilon$ has zero mean and variance $\sigma^2$.   For $m$ and $m'$, we have used the fixed values $m=6$, $m'=2$.

\paragraph{Model 7:} $\boldsymbol{Y=\sin(2\pi X)+\varepsilon}$\\
with $X\sim \mbox{equal}(m,-1,1)$ and $\varepsilon\sim -\sigma\sqrt{m'} + \frac{2\sigma}{\sqrt{m'}}\mbox{binom}(m',0.5)$.

This is a discrete version of Model 1. $\sigma$ is the standard deviation of the noise $\varepsilon$. For $m$ and $m'$, we have used the fixed values $m=6$, $m'=2$.

\paragraph{Models 8-10:} $\boldsymbol{Y,X}$ {\bf independent}

Here we examined one continuous and two discrete examples:
\begin{description}\itemsep0ex
  \item[8)] $X\sim \mbox{unif}(-1,1)$ and $Y\sim \mbox{unif}(-1,1)$
  \item[9)] $X\sim \mbox{equal}(m,-1,1)$ and $Y\sim \mbox{equal}(m',-1,1)$
  \item[10)] $X\sim \mbox{binom}(m,p)$ and $Y\sim \mbox{binom}(m',p')$
\end{description}
We have used $m=3, p=0.5$ and $m'=6, p'=0.3$

\begin{figure*}
  \centering
  \subfigure[model 1: $Y\sim X + \varepsilon$ (continuous)]{
      \includegraphics[width=0.95\columnwidth]{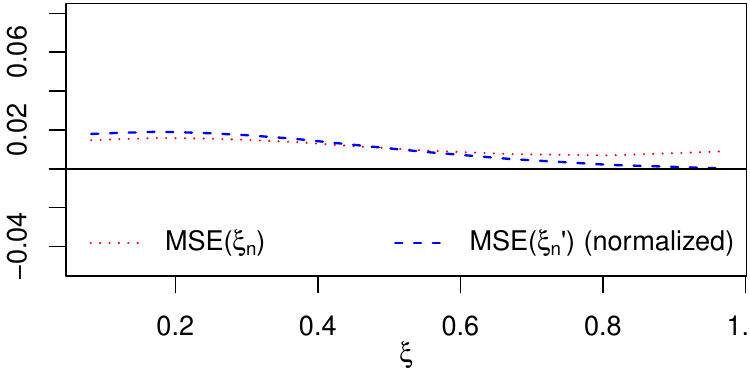}
      \label{fig:mse:model1}
  }
  \subfigure[model 2: $Y\sim X^2 + \varepsilon$ (continuous)]{
      \includegraphics[width=0.95\columnwidth]{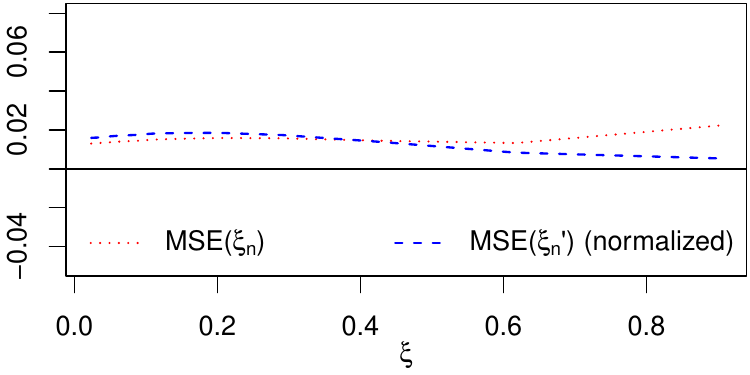}
      \label{fig:mse:model6}
  }
  \subfigure[model 3: $Y\sim \sin(2\pi X) + \varepsilon$ (continuous)]{
      \includegraphics[width=0.95\columnwidth]{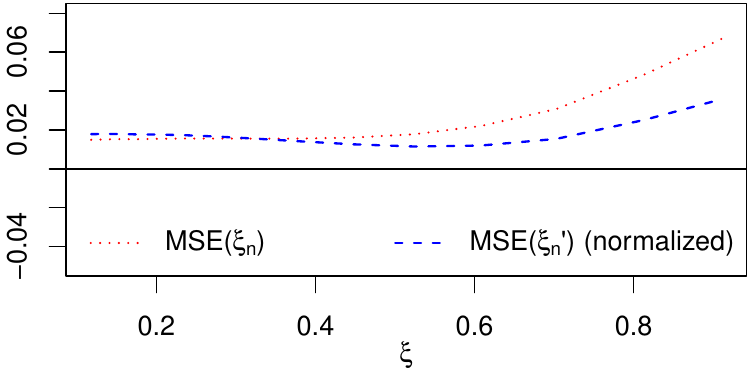}
      \label{fig:mse:model3}
  }
  \subfigure[model 7: $Y\sim \sin(2\pi X) + \varepsilon$ (discrete)]{
      \includegraphics[width=0.95\columnwidth]{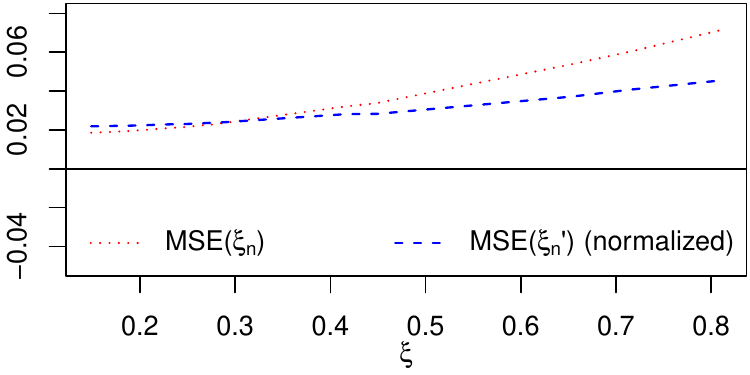}
      \label{fig:mse:model7}
  }
  \caption{\label{fig:mse} MSE of $\xi_n$ and the normalized $\xi_n'$ as a function of the true value of $\xi$ for different models and $n=30$.}
\end{figure*}

\subsection{Effect on bias}
\label{sec:simulation:bias}
For $\xi=0$, i.e., for independence of $X$ and $Y$, which we simulated in models 8-10, already $\xi_n$ itself showed to be unbiased. As $\xi'_n$ is only a scaled version of $\xi_n$, it was unbiased, too, in these cases.

For $\xi>0$, however, the normalization reduced the bias in all cases. As the examples in Fig.~\ref{fig:bias} show, the normalization made the estimator almost unbiased even for small $n$ for models with a monotonous relationship between $X$ and $Y$, as in models 1 and 4. We observed the same for model 5 (not shown in Fig.~\ref{fig:bias}). For more complicated relationships, as in models 2, 3, 6, and 7, the bias reduction was not as strong, but nevertheless quite distinct. This was to be expected because, even for strict functional relationship $Y=f(X)$, $\xi_n$ takes its maximum value only for monotonous $f$ (see Theorem \ref{theorem:maxxi}). Normalizing by the maximum value thus can only partially correct for this bias.

\subsection{Effect on MSE}
\label{sec:simulation:mse}
As the normalization scales $\xi_n$ by a factor greater than one, it increases its variance. It thus happened that, for small $\xi$, the increase in variance outweighed the bias reduction for some models. Our models 1-3 and 5-7 allowed for covering the full range of $\xi$ by varying the parameter $\sigma$. The resulting MSE and its relation to the true value for $\xi$ for some fixed value of $n$ is shown in Fig.~\ref{fig:mse}. As can be seen, the normalization typically reduced the MSE if $\xi\gtrsim 0.4$ and otherwise slightly increased it. The MSE reduction for high $\xi$ was greater, however, than its increase for small $\xi$.

\section{Confidence intervals for $\xi$}
\label{sec:coverage}


Basically, there are two different approaches to construct a confidence interval for $\xi$ on basis of $\xi_n$:
\begin{enumerate}
\item the normal approximation interval $\xi_n\pm z_{1-\alpha/2} \hat{\sigma}_{\xi_n}$ based on some estimator $\hat{\sigma}_{\xi_n}^2$ for $\operatorname{Var}(\xi_n)$
\item a non-parametric interval based on a bootstrap estimation of the distribution of $\xi_n$
\end{enumerate}
That the normality assumption required for the first approach is asymptotically justified for independent $X$ and $Y$ was already proven by Chatterjee \cite{chatterjee21}. Moreover, Lin \& Han have proven its asymptotic validity in the continuous case, and they even presented a consistent estimator for the variance in this situation \cite[theorem 1.2]{dette23}. For the discrete case, asymptotic normality of $\xi_n$ has not yet been proven and no estimator for $\operatorname{Var}(\xi_n)$ is known, yet Dette \& Kroll \cite{dette23} have suggested to use the m-out-of-n bootstrap for estimating the variance and to construct confidence intervals on assuming normality. They have used the m-out-of-n bootstrap, because the usual n-out-of-n bootstrap has been demonstrated to fail in estimating $\hat{\sigma}_{\xi_n}$ \cite{lin23}.

\begin{figure*}
  \centering
  \subfigure[$X, Y\sim \operatorname{unif}(-1,1)$]{
      \includegraphics[width=0.95\columnwidth]{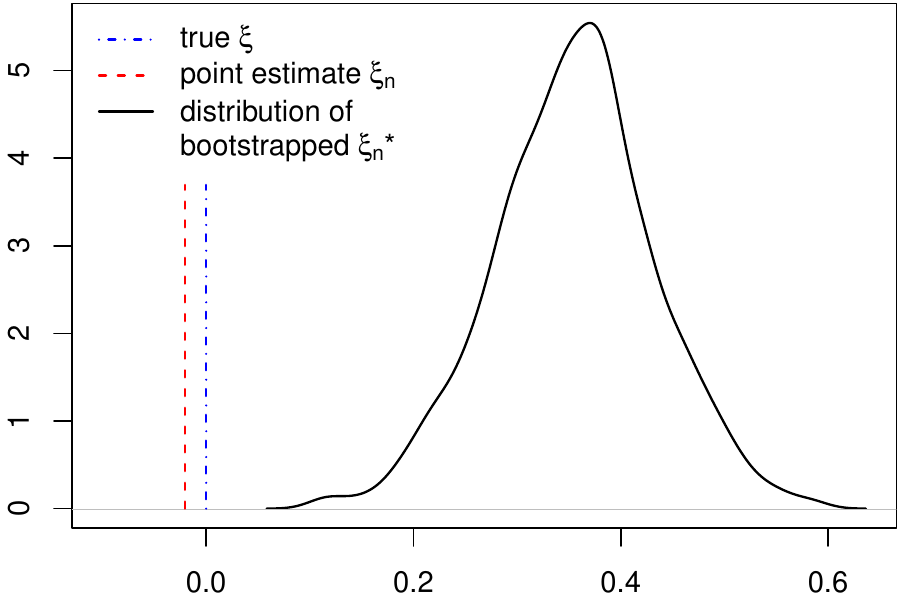}
      \label{fig:bootstrap:continuous}
  }
  \subfigure[$X, Y\sim \operatorname{equal}(5,-1,1)$]{
    \includegraphics[width=0.95\columnwidth]{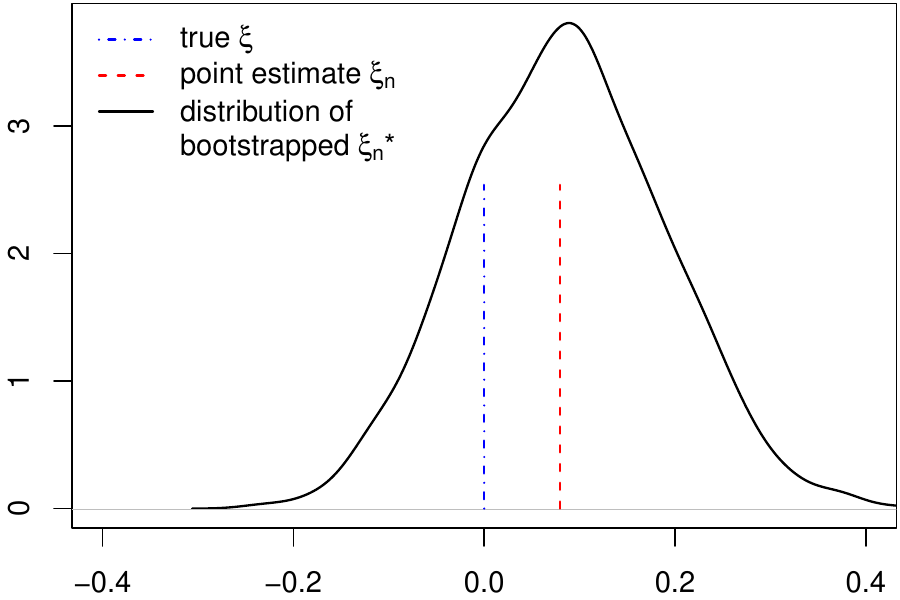}
      \label{fig:bootstrap:discrete}
  }
  \caption{\label{fig:bootstrap} Typical distribution (kernel density plot) of $\xi_n(\vec{x},\vec{y})$ for $n=50$, computed from n-out-of-n bootstrap samples $(\vec{x},\vec{y})^*$ for independent continuous (left) and discrete (right) $X$ and $Y$.}
\end{figure*}

The second approach has not yet been studied in the literature, although it circumvents the question whether normality asymptotically holds. In the present study, we implemented and evaluated it, and also compared it to the method by Dette \& Kroll. At first, we had tried the usual n-out-of-n bootstrap, which has been reported to provide non-parametric confidence intervals with decent coverage probability for a wide range of estimators \cite{diciccio96,dalitz17}. With $R=1000$ bootstrap repetitions, simulating a million cases was no longer feasible and we reduced the number of test cases to $N=10^5$. For all models with continuous $Y$, it turned out, however, that the bootstrap confidence intervals did not even include the point estimate $\xi_n$, which lead to coverage probabilities close to zero. This even occurred for {\em independent} $X$ and $Y$, as can be seen in Fig.~\ref{fig:bootstrap:continuous}. A possible explanation for this effect is that the bootstrap samples are not representative for the distribution of $Y$ because repeated drawing with replacement results in a discrete distribution with many multiple values, whereas in the original distribution of $Y$ a value almost never occurs multiple times. This means that the denominator in Eq.~(\ref{eq:xin}) is no longer $n(n^2-1)/3$ as it is for continuous $Y$. This explanation is in agreement with the absence of this effect for discretely distributed $Y$ (see Fig.~\ref{fig:bootstrap:discrete}), although the coverage probability of the n-out-of-n bootstrap was unsatisfactory even in the discrete cases: For models 5-7 with $\sigma=0.1$ and sample size $n=2000$, the coverage probabilities of the usually best performing BCa bootstrap interval \cite{diciccio96,davison97} were only about 0.65 even for the normalized estimator $\xi_n'$, which was considerably less than the nominal value 0.9.

We therefore resorted to non-parametric intervals based on the m-out-of-n bootstrap, which has been shown to work in cases where the usual bootstrap fails \cite{bickel97}. For our problem, it is particularly appropriate because it allows for drawing {\em without replacement} and thus does not misrepresent a continuous distribution as a discrete  distribution, as it happens by repeated drawing with replacement in the n-out-of-n bootstrap. Politis \& Romano \cite{politis94} have shown that the m-out-of-n bootstrap requires knowledge of an appropriate scaling factor $\tau_n$. With the aid of this factor, the quantiles $q(1-\alpha/2)$ and $q(\alpha/2)$ of the scaled bootstrap distribution $\tau_m(\xi_m^* - \xi_n)$ are computed, where $\xi_m^*$ denotes the bootstrap samples obtained by m-fold drawing without replacement, and the confidence interval is estimated as
\begin{equation}
  \label{eq:ci-m-out-of-n}
  \left[ \xi_n - \frac{q(1-\alpha/2)}{\tau_n},\; \xi_n - \frac{q(\alpha/2)}{\tau_n} \right]
\end{equation}

In the subsequent article \cite{bertail99}, Bertail, Politis \& Romano explained that $\tau_n$ must be chosen such that $\tau_n^2\operatorname{Var}(\xi_n)$ converges to some constant $V$. For the special case $\xi=0$, Chatterjee already proved this convergence for $\tau_n^2=n$ \cite{chatterjee21}. In the simulations of our models, we observed the same relationship $\operatorname{Var}(\xi_n)\sim n^{-1}$, as can be seen in the examples in Fig.~\ref{fig:var-n}. The results for the continuous models 1-3 are in agreement with Lin \& Han's \cite{lin22} proof of root-n consistency of $\xi_n$ in the continuous case. Similar results were obtained for $\xi_n'$. We therefore set the scaling factor to $\tau_n=\sqrt{n}$.

\begin{figure}[b!]
  \centering
  \includegraphics[width=1.0\columnwidth]{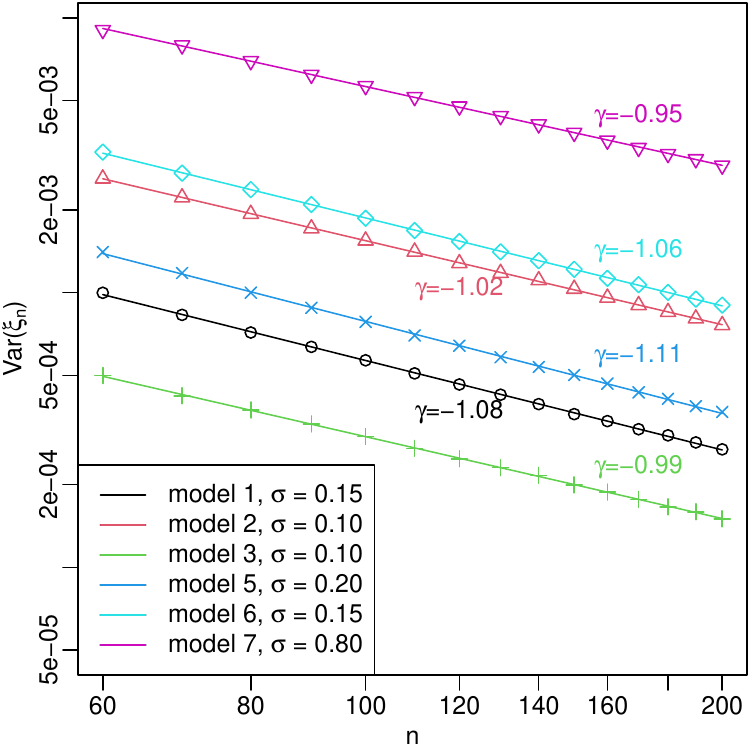}
  \caption{\label{fig:var-n} Logarithmic plot of the variation of $\operatorname{Var}(\xi_n)$ with $n$. The lines have been fitted with the model $\log \operatorname{Var}(\xi_n) = \log V + \gamma \log n$.}
\end{figure}

\begin{figure*}
  \centering
  \subfigure[model 1: $Y\sim X + \varepsilon$ (continuous)]{
      \includegraphics[width=0.95\columnwidth]{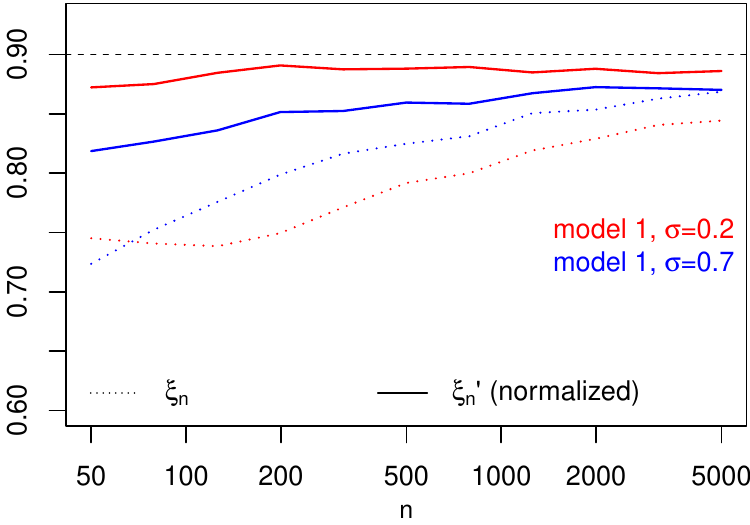}
      \label{fig:Pcov:model1}
  }
  \subfigure[model 5: $Y\sim X + \varepsilon$ (discrete)\newline
    \hspace*{1.6em}model 3: $Y\sim \sin(2\pi X) + \varepsilon$ (continuous)]{
      \includegraphics[width=0.95\columnwidth]{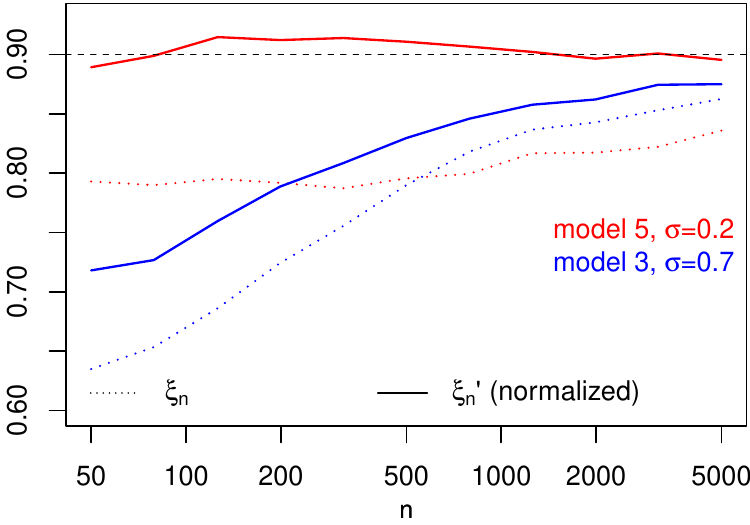}
      \label{fig:Pcov:model35}
  }
  \caption{\label{fig:Pcov} Coverage probability of 90\% m-out-of-n bootstrap confidence intervals of the raw $\xi_n$ and the normalized $\xi_n'$.}
\end{figure*}

Another inconvenience of the m-out-of-n bootstrap is that it has a parameter $m$ which has to be chosen. For asymptotic convergence of the bootstrap distribution to the true distribution, Politis \& Romano \cite{politis94} proved the sufficient conditions $m\to\infty$ and $m/n\to 0$ for $n\to\infty$, but this still leaves a wide range of options, e.g.~$m=c\cdot n^{\alpha}$ with $0<\alpha<1$. Moreover, they presented examples showing that the convergence rate of the coverage probability to the nominal value depends on the choice for $m(n)$.

\lstset{language=R,
  basicstyle=\small \ttfamily,
  literate={.help}{.help}5,
  keywordstyle=\ttfamily,
  frame=bottomline,
  floatplacement=!t,
  aboveskip=0pt,
  belowskip=0pt,
  captionpos=b
}
\begin{lstlisting}[float, caption={R implementation for computing the non-parametric m-out-of-n bootstrap confidence interval for the normalized estimator $\xi_n'$.}, label=lst:ci]
# wrapper around xicor() for bootstrap
xiboot <- function(indices, data) {
  xicor(data[indices,"x"],
        data[indices,"y"]) /
        xicor(data[indices,"y"],
              data[indices,"y"])
}

# confidence interval after
# Politis & Romano (1994)
m.out.of.n.ci <- function(
     data, conf=0.90, R=1000)
{
  # point estimate
  xi.n <- xiboot(1:n, data)

  # bootstrap distribution
  m <- round(2*sqrt(n))
  indices <- replicate(R,
      sample(1:n, size=m, replace=F))
  xi.star <- apply(indices,
      MAR=2, xiboot, data=data)

  # confidence interval
  tau <- sqrt
  xq <- quantile(tau(m)*(xi.star-xi.n),
           c((1+conf)/2, (1-conf)/2))
  return(c(xi.n - xq[1] / tau(n),
           xi.n - xq[2] / tau(n)))
}
\end{lstlisting}

To circumvent this problem, different procedures for a data driven choice of $m$ have been suggested in the literature \cite{goetze01,chung01,bickel08}. We have tried the method by G{\"o}tze and Ra{\v{c}}kauskas \cite{goetze01}, which minimizes the distance between the empirical cumulative distribution functions of the (scaled) statistic under consideration for the two m-out-of-n bootstraps with $m$ and $m/2$. As a distance measure, we have tried both the Kolmogorov distance  and the $L_2$ distance, and we implemented a Golden Section Search \cite{press92} for efficiently finding the minimum. For all models, however, the resulting coverage probability was less than 0.8 over the range $50\leq n \leq 5000$.

\begin{figure*}
  \centering
  \subfigure[model 1: $Y\sim X + \varepsilon$ (continuous)\newline
  \hspace*{1.6em}model 3: $Y\sim X^2 + \varepsilon$ (continuous)]{
      \includegraphics[width=0.95\columnwidth]{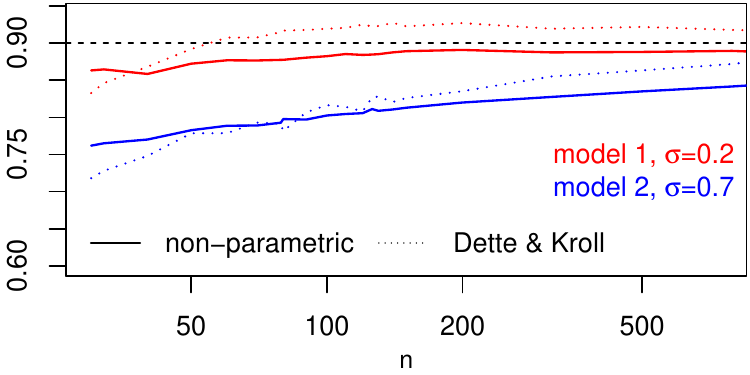}
      \label{fig:bias:model12}
  }
  \subfigure[model 5: $Y\sim X + \varepsilon$ (discrete)\newline
    \hspace*{1.6em}model 7: $Y\sim \sin(2\pi X) + \varepsilon$ (discrete)]{
      \includegraphics[width=0.95\columnwidth]{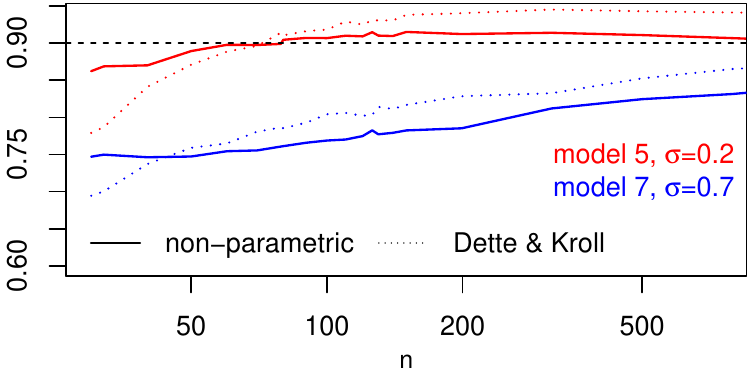}
      \label{fig:dk:model57}
  }
  \subfigure[model 4: $Y\sim XY$ (discrete)]{
      \includegraphics[width=0.95\columnwidth]{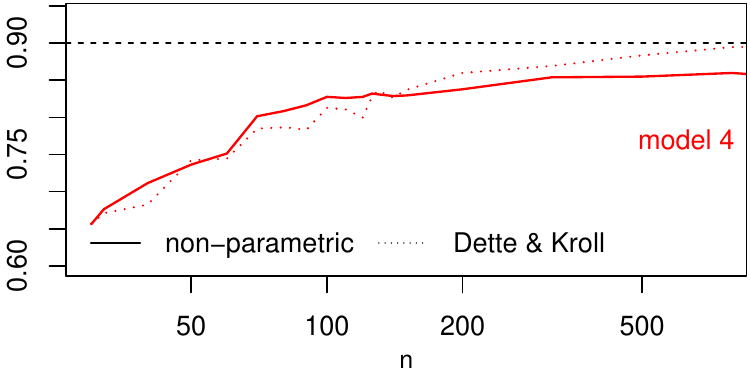}
      \label{fig:dk:model4}
  }
  \subfigure[model 9: $Y,X$ uniform and independent (discrete)]{
      \includegraphics[width=0.95\columnwidth]{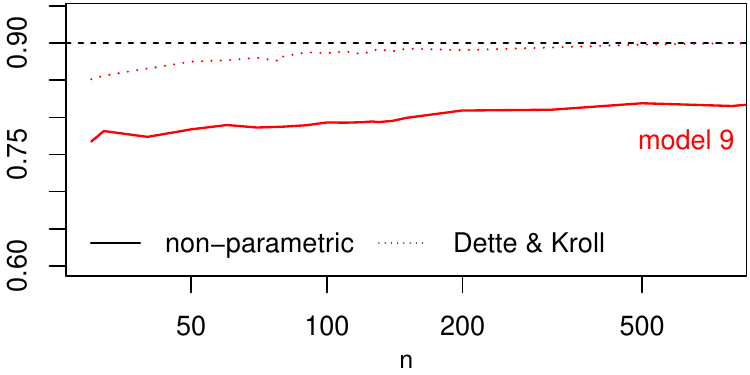}
      \label{fig:dk:model9}
  }
  \caption{\label{fig:dk} Comparison of the coverage probability of the 90\% non-parametric confidence interval and the $\pm z_{1-\alpha/2}\hat{\sigma}_{\xi_n}$ interval after Dette \& Kroll \cite{dette23}.}
\end{figure*}

We therefore tried different choices of the form $m=cn^\alpha$ with $c\in\{1,2,3,4\}$ and $\alpha\in\{1/2,2/3,3/4\}$. Among these combinations, the highest coverage probabilities for $500\leq n\leq 5000$ were achieved with $c=2$ and $\alpha=1/2$. We therefore made the choice $m=2\sqrt{n}$, and the code for computing the confidence intervals is given in Listing \ref{lst:ci}. As can be seen in Fig.~\ref{fig:Pcov}, both the resulting coverage probability and the speed of the convergence to the nominal value was different for the different models. In all cases, however, it approached the nominal value for large $n$, albeit rather slow for some models, e.g. for model 3 with $\sigma=0.7$. In all cases, the coverage probability was higher for the normalized estimator $\xi_n'$ than for the original estimator $\xi_n$. Moreover, for higher values of $\xi$, which correspond to smaller values for $\sigma$ in Fig.~\ref{fig:Pcov}, the difference of the coverage probability was greater than for smaller values of $\xi$ (higher $\sigma$), and the confidence intervals for the normalized $\xi_n'$ approached the nominal value 90\% for much smaller $n$ compared to the use of the unnormalized $\xi_n$.

To compare the non-parametric confidence intervals with the asymptotic normal $\pm z_{1-\alpha/2}\hat{\sigma}_{\xi_n}$ intervals, we have also estimated $\hat{\sigma}_{\xi_n}$ with an m-out-of-n bootstrap with $m=\sqrt{n}$ as suggested by Dette \& Kroll \cite{dette23}. The convergence of the coverage probability of the asymptotic normal interval to the nominal value $1-\alpha=0.9$ for $n\to\infty$ turned out to be faster in all of our simulated cases, as is exemplified in Fig.~\ref{fig:dk}. In some models with high $\xi$ (coresponding to small $\sigma$ in our models 1\&5), the coverage probability of the normal based intervals even exceeded the nominal value and approached it from above as $n\to\infty$. For small $n$, there is typically a greater deviation of the bootstrap distribution from normality, so that the non-parametric intervals showed a higher coverage probability for many models in the range $n\lesssim 50$, as can be seen in Figs.~\ref{fig:bias:model12} \& \ref{fig:dk:model57}. This was not universal, though: for model 4, the coverage probabilities were similar for both methods (see Fig.~\ref{fig:dk:model4}), and, for models 9\&10, the normal interval had higher coverage probability throughout (see Fig.~\ref{fig:dk:model9}).

\section{Conclusions}
\label{sec:conclusions}
Our simulations show that scaling Chatterjee's rank correlation coefficient $\xi_n$ by its upper bound both reduces its bias and improves the coverage probability of a non-parametric m-out-of-n bootstrap confidence interval. We have observed that the usual n-out-of-n bootstrap does not work, and our utilization of the m-out-of-n bootstrap without replacement, also known as ``n choose m bootstrap'', is a viable approach to obtain asymptotically valid confidence intervals. For some of our simulated models, however, a nearly nominal coverage probability was achieved only for fairly high sample sizes $n\geq 1000$. The method by Dette \& Kroll \cite{dette23} based on a normality assumption showed an earlier convergence to the nominal coverage probability and is thus recommended, unless the sample size is very small, e.g. about 50 and below, in which case the non-pramateric intervals had higher coverage probability in most cases.

\section*{Acknowledgements}
\label{sec:acknowledgements}
We thank the anonymous reviewers for their valuable comments. The first version of this article has been revised after we have learnt about the preprint of \cite{dette23} and the references \cite{lin22,lin23} given therein. We thank Felix L\"{o}gler for doing some of the additional simulations that thereby became necessary to compare the different methods for constructing a confidence interval.

\section*{Appendix: Computation of $\xi$}
\label{appendix}
For estimating the bias of $\xi_n$ in the different models, it is necessary to compute the asymptotic value $\xi$. As $\xi_n$ is known to be asymptotically unbiased, the asymptotic value can be approximated by a Monte Carlo simulation with a large sample size. Alternatively, it can be done with Eq.~(\ref{eq:xicondprob}) by means of symbolic integration, numeric integration, or a combination of both. In this appendix, the latter approach is chosen.

\paragraph{Model 1:} $\boldsymbol{Y=X+\varepsilon}$ \\
with $X\sim \mbox{unif}(a,b)$ and $\varepsilon\sim \mbox{norm}(0,\sigma^2)$.

In this case, the probability density of $X$ is $\mbox{dunif}(a,b)(x)= 1/(b-a)$ for $x\in[a,b]$, and the probability density of $\varepsilon$ is $\mbox{dnorm}(0,\sigma^2)(x)= \varphi(z/\sigma)/\sigma$, where $\varphi(z)=e^{-z^2/2}/\sqrt{2\pi}$ is the density of the standard normal distribution.

As $Y$ is the sum of two independent variables, its probability density $f_Y(y)$ is the convolution of the respective densities, i.e.,
\begin{align}
  f_Y(y) & = \mbox{dunif(a,b)}*\mbox{dnorm}(0,\sigma^2) (y) \nonumber \\
  & = \frac{1}{b-a}\left( \Phi\left(\frac{y-a}{\sigma}\right)
  - \Phi\left(\frac{y-b}{\sigma}\right) \right)
\end{align}
where $\Phi$ is the cumulative distribution function of the standard normal distribution. Moreover,
\begin{align}
  P(Y\geq t|X\! =\! x) & = \int\limits_t^\infty \frac{1}{\sigma}\,\varphi\left(\frac{y-x}{\sigma}\right) dy \nonumber \\
  & = \Phi\left(\frac{x-t}{\sigma}\right)
\end{align}
The inner integral in Theorem \ref{theorem:xicondprob.cont} with respect to $d\lambda(x)=\mbox{dunif}(a,b)(x)\,dx$ can also be symbolically integrated to \cite{owen80}
\begin{multline}
  \frac{1}{b-a}\int\limits_a^b P(Y\!\geq\! t|X\!=\!x)^2 dx 
  = \frac{\sigma}{b-a} \bigg[ z\Phi(z)^2 \\
    + 2\Phi(z)\varphi(z)
   -\frac{1}{\sqrt{\pi}}\Phi(z\sqrt{2})\bigg]_{\frac{a-t}{\sigma}}^{\frac{b-t}{\sigma}}
\end{multline}
Eventually, we have evaluated the remaining integral in Theorem \ref{theorem:xicondprob.cont} with respect to $d\mu(t)=f_Y(t)dt$ numerically with the R function {\em integrate()}. We used this mostly symbolic solution as a ground truth for testing the solely numeric integration that was necessary for the other models.

\paragraph{Models 2 \& 3:} $\boldsymbol{Y=f(X)+\varepsilon}$ \\
with $X\sim \mbox{unif}(a,b)$ and $\varepsilon\sim \mbox{norm}(0,\sigma^2)$.

Both models can be written in this from with $f(x)=x^2$ for model 2, and $f(x)=\sin(x)$ for model 3. In this more general situation, it is
\begin{align}
  P(Y\!\geq\! t|X\! =\! x) & = \int\limits_t^\infty \frac{1}{\sigma}\,\varphi\left(\frac{y-f(x)}{\sigma}\right) dy \nonumber \\
  & = \Phi\left(\frac{f(x)-t}{\sigma}\right)
\end{align}
and the other integrals cannot be symbolically solved. We therefore computed all other integrals with the function from Listing \ref{lst:integratexicor}. It should be noted that the convolution integral {\em py} in Listing \ref{lst:integratexicor} might be computed more efficiently with an FFT based approach \cite{fftconv14}, but the computation by means of {\em integrate()} was fast enough for our purpose.

\paragraph{Model 4:} $\boldsymbol{Y=XZ}$ \\
with $X\sim \mbox{bernoulli}(p)$ and $Z\sim \mbox{bernoulli}(p')$.

This is the model for which Chatterjee already gave a formula for $\xi$ \cite{chatterjee21}. $Y$ can only take the two values $0$ and $1$ with probabilities
\begin{equation}
  P(Y=1)=pp' \quad\mbox{and}\quad P(Y=0)=1-pp'
\end{equation}
and we have
\begin{align}
  P(Y\geq 0|X\! =\! x) & = 1  && \quad\mbox{for both }x \\
  P(Y\geq 1|X\! =\! x) & =  p' && \quad\mbox{for }x=1 \nonumber \\
  P(Y\geq 1|X\! =\! x) & = 0  && \quad\mbox{for }x=0 \nonumber
\end{align}
Inserting everything into Eq.~(\ref{eq:xicondprob}) yields
\begin{equation}
  \xi = \frac{(1-p)p'}{1-pp'}
\end{equation}

\lstset{language=R,
  basicstyle=\small \ttfamily,
  literate={.help}{.help}5,
  keywordstyle=\ttfamily,
  frame=bottomline,
  floatplacement=!t,
  aboveskip=0pt,
  belowskip=0pt,
  captionpos=b
}
\begin{lstlisting}[float, caption=R implementation of the numeric integration to compute the value $\xi$ according to Theorem \ref{theorem:xicondprob.cont}., label=lst:integratexicor]
# numeric integration of xi
# for models f(X) + eps
#   px   = density function of x
#   peps = density function of eps
#   PYx  = P(Y>=t|x)
xi.cont.fx <- function(
        f, px, peps, PYx,
        minx=-Inf, maxx=Inf,
        miny=-Inf, maxy=Inf) {
  
  # convolution integral
  py <- function(y) {
    ff <- function(x) {
      px(x)*peps(y-f(x))
    }
    integrate(ff, minx, maxx)$value
  }

  # integral w.r.t. dlambda(x)
  inner.integral <- function(t) {
    ff <- function(x) {
      PYx(t,x)^2 * px(x)
    }
    integrate(ff, minx, maxx)$value
  }

  # integral w.r.t dmu(t)
  ff <- Vectorize( function(t) {
    py(t) * inner.integral(t)
  } )
  6*integrate(ff, miny, maxy)$value - 2
}
\end{lstlisting}

\paragraph{Model 5 \& 6:} $\boldsymbol{Y=f(X)+\varepsilon}$ \\
with $X\sim \mbox{equal}(a,b,n)$ and\\ $\varepsilon\sim -\sigma\sqrt{m} + \frac{2\sigma}{\sqrt{m}}\mbox{binom}(m,0.5)$.

For discrete distributions, all integrals in Eq.~(\ref{eq:xicondprob}) actually are finite sums which can be readily evaluated by the computer. For computing the distribution of $Y$, i.e. the convolution between the distributions of $f(X)$ and $\varepsilon$, we have used the function {\em conv()} from the R package {\em kSamples} \cite{kSamples}.

\onecolumn
\bibliographystyle{ieeetr}
\bibliography{xicor-bias}

\begin{thebibliography}{10}

\bibitem{chatterjee21}
S.~Chatterjee, ``A new coefficient of correlation,'' {\em Journal of the
  American Statistical Association}, vol.~116, no.~536, pp.~2009--2022, 2021.

\bibitem{azadkia21}
M.~Azadkia and S.~Chatterjee, ``A simple measure of conditional dependence,''
  {\em The Annals of Statistics}, vol.~49, no.~6, pp.~3070--3102, 2021.

\bibitem{fischer21}
S.~M. Fischer and M.~A. Lewis, ``A robust and efficient algorithm to find
  profile likelihood confidence intervals,'' {\em Statistics and Computing},
  vol.~31, no.~4, p.~38, 2021.

\bibitem{greene00}
W.~H. Greene, {\em Econometric Analysis}.
\newblock New Jersey: Prentice Hall, 4~ed., 2000.

\bibitem{efron79}
B.~Efron, ``Bootstrap methods: another look at the jackknife,'' {\em Annals of
  Statistics}, vol.~7, no.~1, pp.~1--26, 1979.

\bibitem{diciccio96}
T.~J. DiCiccio and B.~Efron, ``Bootstrap confidence intervals,'' {\em
  Statistical Science}, vol.~11, no.~3, pp.~189--228, 1996.

\bibitem{davison97}
A.~C. Davison and D.~V. Hinkley, {\em Bootstrap Methods and their Application}.
\newblock Cambridge University Press, 1997.

\bibitem{dey16}
A.~K. Dey and K.~P. Das, ``Modeling extreme hurricane damage using the
  generalized {P}areto distribution,'' {\em American Journal of Mathematical
  and Management Sciences}, vol.~35, no.~1, pp.~55--66, 2016.

\bibitem{abadie08}
A.~Abadie and G.~W. Imbens, ``On the failure of the bootstrap for matching
  estimators,'' {\em Econometrica}, vol.~76, no.~6, pp.~1537--1557, 2008.

\bibitem{lin23}
Z.~Lin and F.~Han, ``On the failure of the bootstrap for chatterjee’s rank
  correlation,'' {\em Biometrika}, p.~asae004, 02 2024.

\bibitem{politis94}
D.~N. Politis and J.~P. Romano, ``Large sample confidence regions based on
  subsamples under minimal assumptions,'' {\em The Annals of Statistics},
  vol.~22, no.~4, pp.~2031--2050, 1994.

\bibitem{bickel97}
P.~J. Bickel, F.~G{\"o}tze, and W.~R. van Zwet, ``Resampling fewer than n
  observations: gains, losses, and remedies for losses,'' {\em Statistica
  Sinica}, vol.~7, no.~1, pp.~1--31, 1997.

\bibitem{dette23}
H.~Dette and M.~Kroll, ``A simple bootstrap for {C}hatterjee's rank
  correlation,'' 2023.
\newblock Preprint. \url{https://arxiv.org/abs/2308.01027}.

\bibitem{lin22}
Z.~Lin and F.~Han, ``Limit theorems of {C}hatterjee's rank correlation,'' 2022.
\newblock Preprint. \url{https://arxiv.org/abs/2204.08031}.

\bibitem{dette13}
H.~Dette, K.~F. Siburg, and P.~A. Stoimenov, ``A copula-based non-parametric
  measure of regression dependence,'' {\em Scandinavian Journal of Statistics},
  vol.~40, no.~1, pp.~21--41, 2013.

\bibitem{xicor_r}
S.~Chatterjee and S.~Holmes, {\em XICOR: Robust and generalized correlation
  coefficients}, 2023.
\newblock R package. \url{https://CRAN.R-project.org/package=XICOR}.

\bibitem{dalitz17}
C.~Dalitz, ``Construction of confidence intervals,'' Tech. Rep. 2017-01,
  Hochschule Niederrhein, Fachbereich Elektrotechnik und Informatik, 2017.

\bibitem{bertail99}
P.~Bertail, D.~N. Politis, and J.~P. Romano, ``On subsampling estimators with
  unknown rate of convergence,'' {\em Journal of the American Statistical
  Association}, vol.~94, no.~446, pp.~569--579, 1999.

\bibitem{goetze01}
F.~G{\"o}tze and A.~Ra{\v{c}}kauskas, ``Adaptive choice of bootstrap sample
  sizes,'' {\em Lecture Notes-Monograph Series}, vol.~36 (State of the Art in
  Probability and Statistics), pp.~286--309, 2001.

\bibitem{chung01}
K.-H. Chung and S.~M. Lee, ``Optimal bootstrap sample size in construction of
  percentile confidence bounds,'' {\em Scandinavian Journal of Statistics},
  vol.~28, no.~1, pp.~225--239, 2001.

\bibitem{bickel08}
P.~J. Bickel and A.~Sakov, ``On the choice of m in the m out of n bootstrap and
  confidence bounds for extrema,'' {\em Statistica Sinica}, vol.~18, no.~3,
  pp.~967--985, 2008.

\bibitem{press92}
W.~H. Press, S.~A. Teukolsky, W.~T. Vetterling, and B.~P. Flannery, {\em
  Numerical Recipes in C}.
\newblock Cambridge, USA: Cambridge University Press, 2nd~ed., 1992.

\bibitem{owen80}
D.~B. Owen, ``A table of normal integrals,'' {\em Communications in Statistics
  - Simulation and Computation}, vol.~9, no.~4, pp.~389--419, 1980.

\bibitem{fftconv14}
P.~Ruckdeschel and M.~Kohl, ``General purpose convolution algorithm in {S}4
  classes by means of {FFT},'' {\em Journal of Statistical Software}, vol.~59,
  no.~4, pp.~1--25, 2014.

\bibitem{kSamples}
F.~Scholz and A.~Zhu, {\em kSamples: K-Sample Rank Tests and their
  Combinations}, 2019.
\newblock R package version 1.2-9.
  \url{https://CRAN.R-project.org/package=kSamples}.

\end{thebibliography}

\end{document}